\documentclass[10 pt,journal]{IEEEtran}  
\usepackage{amsthm,amsmath,paralist}
\usepackage{graphicx,color}
\usepackage{supertabular}
\usepackage[colorlinks=true]{hyperref} 
\definecolor{rltblue}{rgb}{0,0,0.75}
\makeatletter
\def\ps@headings{%
\def\@oddhead{\mbox{}\scriptsize\rightmark \hfil \thepage}%
\def\@evenhead{\scriptsize\thepage \hfil \leftmark\mbox{}}%
\def\@oddfoot{}%
\def\@evenfoot{}}
\makeatother 

\newtheorem{theorem}{\textbf{Theorem}}

\newtheorem{proposition}[theorem]{\textbf{Proposition}}
\newtheorem{corollary}[theorem]{\textbf{Corollary}}
\newtheorem{remark}[theorem]{\textbf{Remark}}
\newtheorem{example}[theorem]{\textbf{Example}}
\newtheorem{lemma}[theorem]{\textbf{Lemma}}

\DeclareMathOperator{\tr}{tr}
\DeclareMathOperator{\swt}{swt}
\DeclareMathOperator{\wt}{wt}
\DeclareMathOperator{\spann}{span}

\newcommand{\C}{\mathbf{C}}
\newcommand{\F}{\mathbf{F}}
\newcommand{\Z}{\mathbf{Z}}
\newcommand{\ket}[1]{|#1\rangle}

\newcommand{\scal}[2]{\langle #1\mid #2\rangle_s}

\newcommand{\sdual}{{\perp_s}}
\newcommand{\hdual}{{\perp_h}}

\newcommand{\nix}[1]{}

\begin{document}
\title{Constructions  of Subsystem Codes\\ over Finite Fields}
\author{
Salah A. Aly and Andreas Klappenecker\\
Department  of Computer Science \\
Texas A\&M University, College Station, TX 77843-3112, USA\\ Emails: \{salah,klappi\}@cs.tamu.edu }
\markboth{Aly, Klappenecker:  Constructions  of Subsystem Codes over Finite Fields, 2008.}{Aly, Klappenecker: \MakeLowercase{\textit{}} Constructions  of Subsystem Codes over Finite Fields,  2008.}
\maketitle

\begin{abstract}
Subsystem codes protect quantum information by encoding it in a tensor
factor of a subspace of the physical state space. Subsystem codes
generalize all major quantum error protection schemes, and therefore
are especially versatile. This paper introduces numerous
constructions of subsystem codes. It is shown how one can derive
subsystem codes from classical cyclic codes. Methods to trade the
dimensions of subsystem and co-subystem are introduced that maintain
or improve the minimum distance. As a consequence, many optimal
subsystem codes are obtained.  Furthermore, it is shown how given
subsystem codes can be extended, shortened, or combined to yield new
subsystem codes.  These subsystem code constructions are used to
derive tables of upper and lower bounds on the subsystem code
parameters.
\end{abstract}


\section{Introduction}

Quantum information processing as a growing exciting field has
attracted researchers from different disciplines.  It utilizes the
laws of quantum mechanical operations to perform exponentially
speedy computations. In an open system, one might wonder how to
perform such computations in the presence of decoherence and noise
that disturb quantum states storing quantum information. Ultimately,
the goals of quantum error-correcting codes are to protect quantum
states and to allow recovery of quantum information processed in
computational operations of a quantum computer. Henceforth, one
seeks to design good quantum codes that can be efficiently utilized
for these goals.

\bigskip

A well-known approach to derive quantum error-correcting codes from
 self-orthogonal (or dual-containing) classical codes is called
stabilizer codes, which were introduced a decade ago. The stabilizer
codes inherit some properties of clifford group theory, i.e., they
are stabilized by  abelian finite groups.  In the seminal paper by
Calderbank~\emph{at.
et}~\cite{calderbank98,rains99,shor95}, various
methods of stabilizer code constructions are given, along with their
propagation rules and tables of upper bounds on their parameters. In
a similar tactic, we also present subsystem code structures by
establishing several methods to derive them easily from classical
codes.   Subsystem codes inherit their name from the fact that the
quantum codes are decomposed into two systems as explained in
Section~\ref{sec:background}. The classes of subsystem codes that we
will derive are superior because they can be encoded and decoded
using linear shirt-register operations. In addition, some of these classes turned out to be optimal and MDS codes.

Subsystem codes as we prefer to call them were mentioned in the
unpublished work by Knill~\cite{knill96b,knill06}, in which he
attempted to generalize the theory of quantum error-correcting codes
into subsystem codes. Such codes with their stabilizer formalism
were reintroduced
recently~\cite{bacon06,pre0608,kribs05,kribs05b,poulin05}. An
$((n,K,R,d))_q$ subsystem code is a $KR$-dimensional subspace $Q$ of
$\C^{q^n}$ that is decomposed into a tensor product $Q=A\otimes B$
of a $K$-dimensional vector space $A$ and an $R$-dimensional vector
space $B$ such that all errors of weight less than~$d$ can be
detected by~$A$. The vector spaces $A$ and $B$ are respectively
called the subsystem $A$ and the co-subsystem $B$. For some
background on subsystem codes  see the next section.

\bigskip

 This paper is structured as follows.  In
section~\ref{sec:background}, we present a brief background on
subsystem code structures and present the Euclidean and Hermitian
constructions. In section~\ref{sec:CyclicSubsys}, we derive cyclic
subsystem codes and provide two generic methods of their
constructions from classical cyclic codes. Consequently in
section~\ref{sec:dimensions}, we construct families of subsystem BCH
and RS codes from classical BCH and RS over $\F_q$ and $\F_{q^2}$
defined using their defining sets. In Sections~\ref{sec:MDSsubsys},\ref{sec:extendshortensubsys},\ref{sec:combinesubsys},   we establish various methods of
subsystem code constructions by extending and shortening the code
lengths and combining pairs of known codes, in addition, tables of
upper bounds on subsystem code parameters are given. Finally, the
paper is concluded with a discussion and future research directions
in section~\ref{sec:conclusion}.

\medskip

\noindent {\em Notation.} If $S$ is a set, then $|S|$ denotes the
cardinality of the set $S$. Let $q$ be a power of a prime integer
$p$. We denote by $\F_q$ the finite field with $q$ elements. We use
the notation $(x|y)=(x_1,\dots,x_n|y_1,\dots,y_n)$ to denote the
concatenation of two vectors $x$ and $y$ in $\F_q^n$. The symplectic
weight of $(x|y)\in \F_q^{2n}$ is defined as
$$\swt(x|y)=\{(x_i,y_i)\neq (0,0)\,|\, 1\le i\le n\}.$$ We define
$\swt(X)=\min\{\swt(x)\,|\, x\in X, x\neq 0\}$ for any nonempty
subset $X\neq \{0\}$ of $\F_q^{2n}$.

The trace-symplectic product of two vectors $u=(a|b)$ and
$v=(a'|b')$ in $\F_q^{2n}$ is defined as
$$\langle u|v \rangle_s = \tr_{q/p}(a'\cdot b-a\cdot b'),$$ where
$x\cdot y$ denotes the dot product and $\tr_{q/p}$ denotes the trace
from $\F_q$ to the subfield $\F_p$.  The trace-symplectic dual of a
code $C\subseteq \F_q^{2n}$ is defined as $$C^\sdual=\{ v\in
\F_q^{2n}\mid \langle v|w \rangle_s =0 \mbox{ for all } w\in C\}.$$
We define the Euclidean inner product $\langle x|y\rangle
=\sum_{i=1}^nx_iy_i$ and the Euclidean dual of $C\subseteq \F_{q}^n$
as $$C^\perp = \{x\in \F_{q}^n\mid \langle x|y \rangle=0 \mbox{ for
all } y\in C \}.$$ We also define the Hermitian inner product for
vectors $x,y$ in $\F_{q^2}^n$ as $\langle x|y\rangle_h
=\sum_{i=1}^nx_i^qy_i$ and the Hermitian dual of $C\subseteq
\F_{q^2}^n$ as
$$C^\hdual= \{x\in \F_{q^2}^n\mid \langle x|y \rangle_h=0 \mbox{ for all } y\in
C \}.$$

\section{Background on Subsystem Codes}\label{sec:background}

In this section we give a quick overview of subsystem codes.  We
assume that the reader is familiar the theory of stabilizer codes over
finite fields, see~\cite{calderbank98,ketkar06,rains99} and the
references therein.

\subsection{Errors}
Let $\F_q$ denote a finite field with $q$ elements of
characteristic~$p$. Let $\{\ket{x}\mid x \in \F_q\}$ be a fixed
orthonormal basis of $\C^q$ with respect to the standard hermitian
inner product, called the computational basis.  For $a,b \in \F_q$, we
define the unitary operators $X(a)$ and $Z(b)$ on $\C^q$ by
$$ X(a)\ket{x}=\ket{x+a},\qquad Z(b)\ket{x}=\omega^{\tr(bx)}\ket{x},$$
where $\omega=\exp(2\pi i/p)$ is a primitive $p$th root of unity and
$\tr$ is the trace operation from $\F_q$ to $\F_p$. The set $E = \{
X(a)Z(b)\,|\, a,b\in \F_q\}$ forms an orthogonal basis of the
operators acting on $\C^q$ with respect to the trace inner product,
called the error basis.

The state space of $n$ quantum digits (or qudits) is given by
$\C^{q^n}=\C^q \otimes \C^q \otimes \cdots \otimes \C^q$. An error
basis $\mathbf{E}$ on $\C^{q^n}$ is obtained by tensoring $n$
operators in $E$; more explicitly,
$\textbf{E}=\{X(\mathbf{a})Z(\mathbf{b})\mid
\mathbf{a, b} \in \F_q^n\},$
where
$$
\begin{array}{lcl}
X(\mathbf{a}) &=&X(a_1)\otimes\cdots\otimes X(a_n),\\
Z(\mathbf{b}) &=&Z(b_1)\otimes\cdots\otimes Z(b_n)
\end{array}
$$ for $\mathbf{a}=(a_1,\dots,a_n)\in \F_q^n$ and
$\mathbf{b}=(b_1,\dots,b_n)\in \F_q^n$.  The set $E$ is not closed
under multiplication, whence it is not a group.
The group $\mathbf{G}$ generated by $\mathbf{E}$ is given by
$$ \mathbf{G} = \{
\omega^{c}\textbf{E}=\omega^{c}X(\mathbf{a})Z(\mathbf{b})\,|\,
\mathbf{a, b} \in \F_q^n, c\in \F_p\},$$ and $\mathbf{G}$ is called
the error group of $\C^{q^n}$. The error group is an extraspecial
$p$-group.  The weight of an error in $\mathbf{G}$ is given by the
number of nonidentity tensor components; hence, the weight of
$\omega^cX(\mathbf{a})Z(\mathbf{b})$ is given by the symplectic weight
$\swt(\mathbf{a}|\mathbf{b})$.

\subsection{Subsystem Codes}
An $((n,K,R,d))_q$ subsystem code is a subspace $Q=A\otimes B$ of
$\C^{q^n}$ that is decomposed into a tensor product of two vector
spaces $A$ and $B$ of dimension $\dim A=K$ and $\dim B=R$ such that
all errors in $\mathbf{G}$ of weight less than $d$ can be detected by
$A$. We call $A$ the subsystem and $B$ the co-subsystem.  The
information is exclusively encoded in the subsystem $A$. This yields
the attractive feature that errors affecting co-subsystem $B$ alone
can be ignored.

A particularly fruitful way to construct subsystem codes proceeds by
choosing a normal subgroup $N$ of the error group~$\mathbf{G}$, and
this choice determines the dimensions of subsystem and co-subsystem as
well as the error detection and correction capabilities of the
subsystem code, see~\cite{pre0608}. One can relate the normal subgroup
$N$ to a classical code, namely $N$ modulo the intersection of $N$
with the center $Z(\mathbf{G})$ of $\mathbf{G}$ yields the classical
code $X=N/(N\cap Z(\mathbf{G}))$. This generalizes the familiar case of
stabilizer codes, where $N$ is an abelian normal subgroup.  It is
remarkable that in the case of subsystem codes \textit{any} classical
additive code $X$ can occur. It is most convenient that one can also start
with any classical additive code and obtain a subsystem code, as is
detailed in the following theorem from \cite{pre0608}:

\begin{theorem}\label{th:oqecfq}
Let $C$ be a classical additive subcode of\/ $\F_q^{2n}$ such that
$C\neq \{0\}$ and let $D$ denote its subcode $D=C\cap C^\sdual$. If
$x=|C|$ and $y=|D|$, then there exists a subsystem code $Q= A\otimes
B$ such that
\begin{compactenum}[i)]
\item $\dim A = q^n/(xy)^{1/2}$,
\item $\dim B = (x/y)^{1/2}$.
\end{compactenum}
The minimum distance of subsystem $A$ is given by
\begin{compactenum}[(a)]
\item $d=\swt((C+C^\sdual)-C)=\swt(D^\sdual-C)$ if $D^\sdual\neq C$;
\item $d=\swt(D^\sdual)$ if $D^\sdual=C$.
\end{compactenum}
Thus, the subsystem $A$ can detect all errors in $E$ of weight less
than $d$, and can correct all errors in $E$ of weight $\le \lfloor
(d-1)/2\rfloor$.
\end{theorem}
\begin{proof}
See~\cite[Theorem~5]{pre0608}.
\end{proof}

A subsystem code that is derived with the help of the previous theorem
is called a Clifford subsystem code. We will assume throughout this
paper that all subsystem codes are Clifford subsystem codes. In
particular, this means that the existence of an $((n,K,R,d))_q$
subsystem code implies the existence of an additive code $C\le
\F_q^{2n}$ with subcode $D=C\cap C^\sdual$ such that $|C|=q^nR/K$,
$|D|=q^n/(KR)$, and $d=\swt(D^\sdual - C)$.

A subsystem code derived from an additive classical code $C$ is called
pure to $d'$ if there is no element of symplectic weight less than
$d'$ in $C$. A subsystem code is called pure if it is pure to the
minimum distance $d$. We require that an $((n,1,R,d))_q$ subsystem
code must be pure.

We also use the bracket notation $[[n,k,r,d]]_q$ to write the
parameters of an $((n,q^k,q^r,d))_q$ subsystem code in simpler form.
Some authors say that an $[[n,k,r,d]]_q$ subsystem code has $r$
gauge qudits, but this terminology is slightly confusing, as the
co-subsystem typically does not correspond to a state space of $r$
qudits except perhaps in trivial cases. We will avoid this
misleading terminology. An $((n,K,1,d))_q$ subsystem code is also an
$((n,K,d))_q$ stabilizer code and vice versa.

Subsystem codes can be constructed from the classical codes over
$\F_q$ and $\F_{q^2}$. We recall the Euclidean and Hermitian
constructions from~\cite{aly06c}, which are easy consequences of the
previous theorem.

\begin{lemma}[Euclidean Construction]\label{lem:css-Euclidean-subsys}
If $C$ is a $k'$-dimensional $\F_q$-linear code of length $n$ that has
a $k''$-dimensional subcode $D=C\cap C^\perp$ and $k'+k''<n$, then
there exists an
$$[[n,n-(k'+k''),k'-k'',\wt(D^\perp\setminus C)]]_q$$
subsystem code.
\end{lemma}

\begin{lemma}[Hermitian Construction]\label{lem:css-Hermitina-subsys}
If $C$ is a $k'$-dimensional $\F_{q^2}$-linear code of length $n$ that has a
$k''$-dimensional subcode $D=C\cap C^\hdual$ and $k'+k''<n$, then
there exists an
$$[[n,n-(k'+k''),k'-k'',\wt(D^\hdual \setminus C)]]_q$$
subsystem code.
\end{lemma}
\section{Cyclic Subsystem Codes}\label{sec:CyclicSubsys}

In this section we shall derive subsystem codes from classical
cyclic codes. We first recall some definitions before embarking on
the construction of subsystem codes.  For further details concerning
cyclic codes see for instance~\cite{huffman03} and
\cite{macwilliams77}.

Let $n$ be a positive integer and $\F_q$ a finite field with $q$
elements such that $\gcd(n,q)=1$. Recall that a linear code
$C\subseteq \F_q^n$ is called \textit{cyclic} if and only if
$(c_0,\dots,c_{n-1})$ in $C$ implies that $(c_{n-1},c_0,\dots,
c_{n-2})$ in $C$.

For $g(x)$ in $\F_q[x]$, we write $(g(x))$ to denote the principal
ideal generated by $g(x)$ in $\F_q[x]$. Let $\pi$ denote the vector
space isomorphism $\pi\colon \F_q^n\rightarrow R_n=\F_q[x] /
(x^n-1)$ given by
$$ \pi((c_0,\dots,c_{n-1})) =
c_0+c_1x+\cdots+c_{n-1}x^{n-1}+(x^n-1).$$ A cyclic code $C\subseteq
\F_q^n$ is mapped to a principal ideal $\pi(C)$ of the ring $R_n$.
For a cyclic code $C$, the unique monic polynomial $g(x)$ in
$\F_q[x]$ of the least degree such that $(g(x))=\pi(C)$ is called
the \textit{generator polynomial} of $C$. If $C\subseteq \F_q^n$ is
a cyclic code with generator polynomial $g(x)$, then
$$\dim_{\F_q} C = n-\deg g(x).$$

Since $\gcd(n,q)=1$, there exists a primitive $n^\text{th}$ root of unity
$\alpha$ over $\F_q$; that is, $\F_q[\alpha]$ is the splitting field
of the polynomial $x^n-1$ over $\F_q$. Let us henceforth fix this
primitive $n^\text{th}$ primitive root of unity $\alpha$.  Since the
generator polynomial $g(x)$ of a cyclic code $C\subseteq \F_q^n$ is
of minimal degree, it follows that $g(x)$ divides the polynomial
$x^n-1$ in $\F_q[x]$.  Therefore, the generator polynomial $g(x)$ of
a cyclic code $C\subseteq \F_q^n$ can be uniquely specified in terms
of a subset $T$ of $\{0,\dots,n-1\}$ such that
$$ g(x) = \prod_{t\in T} (x-\alpha^t).$$ The set $T$ is called the
\textit{defining set} of the cyclic code $C$ (with respect to the
primitive $n^\text{th}$ root of unity $\alpha$). Since $g(x)$ is a polynomial in $\F_q[x]$, a defining set is the
union of cyclotomic cosets $C_x$, where
$$C_x = \{ xq^i\bmod n\,|\, i\in \Z, i\ge 0\}, \quad 0\le x<n$$
The following lemma recalls
some well-known and easily proved facts about defining sets
(see e.g.~\cite{huffman03}).

\begin{lemma}\label{lem:definingsets}
Let $C_i$ be a cyclic code of length $n$ over $\F_q$ with defining set
a $T_i$ for $i=1,2$. Let $N=\{0,1,\dots,n-1\}$ and $T_1^{a}=\{at\bmod
n\,|\, t\in T\}$ for some integer $a$. Then \begin{compactenum}[i)]
\item $C_1\cap C_2$ has defining set $T_1 \cup T_2$.
\item $C_1+C_2$ has defining set $T_1 \cap T_2$.
\item $C_1 \subseteq C_2$ if and only if $T_2 \subseteq T_1$.
\item $C_1^\perp$ has defining set $N \setminus T_1^{-1}$.
\item $C_1^\hdual$ has defining set $N\setminus T_1^{-r}$ provided
that $q=r^2$ for some positive integer $r$.
\end{compactenum}
\end{lemma}
\textit{Notation.}  Throughout this section, we denote by $N$ the set
$N=\{0,\dots,n-1\}$.  The cyclotomic coset of $x$ will be denoted by
$C_x$.  If $T$ is a defining set of a cyclic code of length $n$ and
$a$ is an integer, then we denote henceforth by $T^a$ the set
$$T^a = \{ at\bmod n\,|\, t\in T\},$$ as in the previous lemma. We use
a superscript, since this notation will be frequently used in set
differences, and arguably $N\setminus T^{-q}$ is more readable
than $N\setminus -qT$.
\smallskip

Now, we shall give a general construction for subsystem cyclic
codes. We say that a
code $C$ is self-orthogonal if and only if $C\subseteq C^\perp$.
We show that if a classical cyclic code is self-orthogonal,
then one can easily construct cyclic subsystem codes.

\begin{proposition}\label{lem:cyclic-subsysI}
Let $D$ be a $k$-dimensional self-orthogonal cyclic code of length $n$
over $\F_q$. Let $T_D$ and $T_{D^\perp}$ respectively denote the
defining sets of $D$ and $D^\perp$. If $T$ is a subset of $T_D
\setminus T_{D^\perp}$ that is the union of cyclotomic cosets, then
one can define a cyclic code $C$ of length $n$ over $\F_q$ by the
defining set $T_C= T_D \setminus (T \cup T^{-1})$.  If $r=|T\cup
T^{-1}|$ is in the range $0\le r< n-2k$, and $d= \min \wt(D^\perp
\setminus C)$, then there exists a subsystem code with parameters
$[[n,n-2k-r,r,d]]_q$.
\end{proposition}
\begin{proof}
Since $D$ is a self-orthogonal cyclic code, we have $D\subseteq
D^\perp$, whence $T_{D^\perp} \subseteq T_{D}$ by
Lemma~\ref{lem:definingsets}~iii).  Observe that if $s$ is an
element of the set  $S= T_D\setminus T_{D^\perp} = T_D \setminus
(N\setminus T_D^{-1})$, then $-s$ is an element of $S$ as well.
In particular, $T^{-1}$ is a subset of $T_D\setminus T_{D^\perp}$.

By definition, the cyclic code $C$ has the defining set $T_C= T_D
\setminus (T \cup T^{-1})$; thus, the dual code $C^\perp$ has the
defining set
$$T_{C^\perp}=N\setminus T_{C}^{-1} =
T_{D^\perp}\cup (T\cup T^{-1}).$$ Furthermore, we have
$$T_C \cup T_{C^\perp}=(T_D \setminus (T \cup T^{-1})) \cup (T_{D^\perp}\cup
T\cup T^{-1})=T_D;$$
therefore, $C \cap C^{\perp}=D$ by Lemma~\ref{lem:definingsets}~i).

Since $n-k=|T_D|$ and $r=|T\cup T^{-1}|$, we have $\dim_{\F_q}
D=n-|T_D| = k$ and $\dim_{\F_q} C = n-|T_C|=k+r$.
Thus, by Lemma~\ref{lem:css-Euclidean-subsys} there exists an
$\F_q$-linear subsystem code with parameters
$[[n,\kappa,\rho,d]]_q$, where
\begin{compactenum}[i)]
\item $\kappa = \dim D^\perp -\dim C=n-k-(k+r)=n-2k-r$,
\item $\rho = \dim C -\dim D= k+r-k=r$,
\item $d= \min \wt(D^\perp \setminus C)$,
\end{compactenum}
as claimed.
\end{proof}

We can also derive subsystem codes from cyclic codes over $\F_{q^2}$
by using cyclic codes that are self-orthogonal with respect to the
Hermitian inner product.

\begin{proposition}\label{lem:cyclic-subsysII}
Let $D$ be a cyclic code of length $n$ over $\F_{q^2}$ such that
$D\subseteq D^\hdual$. Let $T_D$ and $T_{D^\hdual}$ respectively be
the defining set of $D$ and $D^\hdual$. If $T$ is a subset of $T_D
\setminus T_{D^\hdual}$ that is the union of cyclotomic cosets, then
one can define a cyclic code $C$ of length $n$ over $\F_{q^2}$ with
defining set $T_C= T_D \setminus (T \cup T^{-q})$.  If $n-k=|T_D|$ and
$r=|T\cup T^{-q}|$ with $0\le r<n-2k$, and $d=\wt(D^\hdual\setminus
C)$, then there exists an $[[n,n-2k-r,r, d]]_q$ subsystem code.
\end{proposition}
\begin{proof}
Since $D \subseteq D^\hdual$, their defining sets satisfy
$T_{D^\hdual} \subseteq T_{D}$ by
Lemma~\ref{lem:definingsets}~iii). If $s$ is an element of
$T_{D}\setminus T_{D^\hdual}$, then one easily verifies that
$-qs \pmod n$ is an element of $T_{D}\setminus T_{D^\hdual}$.

Let $N=\{0,1,\dots,n-1\}$. Since the cyclic code $C$ has the defining
set $T_C= T_D \setminus (T \cup T^{-q})$, its dual code $C^\hdual$ has
the defining set
$T_{C^\hdual}= N\setminus T_C^{-q} =
T_{D^\hdual}\cup (T\cup T^{-q}).$
We notice that
$$T_C \cup T_{C^\hdual}=(T_D \setminus
(T \cup T^{-q})) \cup (T_{D^\hdual}\cup T\cup T^{-q})=T_D;$$
thus, $C \cap C^{\hdual}=D$ by
Lemma~\ref{lem:definingsets}~i).

Since $n-k=|T_D|$ and $r=|T\cup T^{-q}|$, we have $\dim D=n-|T_D|=k$
and $\dim C=n-|T_C|=k+r$. Thus, by
Lemma~\ref{lem:css-Hermitina-subsys} there exists an $[[n,\kappa,\rho,d]]_q$
subsystem code with
\begin{compactenum}[i)]
\item $\kappa = \dim D^\hdual -\dim C=(n-k)-(k+r)=n-2k-r$,
\item $\rho=\dim C -\dim D= k+r-k=r$,
\item $d= \min \wt(D^\hdual \setminus C)$,
\end{compactenum}
as claimed.
\end{proof}

We include an example to illustrate the construction given in the
previous proposition.
\begin{example}
Consider the narrow-sense BCH code $D^\hdual$ of length $n=31$ over
$\F_4$ with designed distance 5. The defining set $T_{D^\hdual}$ of
$D^\hdual$ is given by $T_{D^\hdual} = C_1\cup C_2\cup C_3\cup C_4 =
C_1 \cup C_3$, where the cyclotomic cosets of $1$ and $3$ are given by
$$C_1=\{1,4,16,2,8\} \quad \text{ and } \quad C_3=\{3,12,17,6,24\}.$$
If $N=\{0,1,\dots,30\}$, then the defining set of the dual code $D$ is
given by $T_D = N\setminus (C_{15} \cup C_7) = C_0\cup C_1\cup C_3\cup
C_5\cup C_{11}. $ Therefore, $D\subset D^\hdual$, $\dim D^\hdual=21$
and $\dim D=10$.  If we choose $T=C_5$, then $T^{-2} = C_{11}$, whence
the defining set $T_C$ of the code $C$ is given by $T_C = T_D\setminus
(C_5 \cup C_{11}) = C_0 \cup C_1\cup C_3.$ It follows that $\dim C =
20$ and $\dim C^\hdual = 11$. Therefore, the construction of the
previous proposition yields a BCH subsystem code with parameters
$[[31,1,10,\geq 5]]_2.$
\end{example}

The general principle behind the previous example yields the following
simple recipe for the construction of subsystem codes: Choose a cyclic
code (such as a BCH or Reed-Solomon code) with known lower bound
$\delta$ on the minimum distance that contains its (hermitian) dual
code, and use Proposition~\ref{lem:cyclic-subsysI} (or
Proposition~\ref{lem:cyclic-subsysII}) to derive subsystem codes.
This approach allows one to control the minimum distance $d$ of the
subsystem code, since $d\ge \delta$ is guaranteed. Another advantage
is that one can exploit the cyclic structure in encoding and decoding
algorithms.

For example, if we start with primitive, narrow-sense BCH codes, then
Proposition~\ref{lem:cyclic-subsysI} yields the following family of
subsystem codes:

\begin{corollary}
Consider a primitive, narrow-sense BCH
code of length $n=q^m-1$ with $m\ge 2$ over $\F_q$ with designed
distance $\delta$ in the range
\begin{equation}\label{eq:ddistrange}
2\le \delta \le q^{\lceil m/2\rceil}-1-(q-2)[m \text{ is odd}].
\end{equation}
If $T$ is a subset of $N \setminus \big(\bigcup_{a=1}^{\delta-1} (C_a\cup C_{-a})\big)$ that is a union
of cyclotomic cosets and $r=|T\cup T^{-1}|$ with $0\le r< n-2k$, where
$k=m\lceil (\delta-1)(1-1/q)\rceil$,  then there exists an
$$ [[q^m-1,q^m-1-2m\lceil (\delta-1)(1-1/q)\rceil -r,r,\ge \delta]]_q$$
subsystem code.
\end{corollary}
\begin{proof}
By \cite[Theorem 2]{aly06a}, a primitive, narrow-sense BCH code
$D^\perp$ with designed distance $\delta$ in the range
(\ref{eq:ddistrange}) satisfies $D\subseteq D^\perp$.  By
\cite[Theorem~7]{aly06a}, the dimension of $D^\perp$ is given by $\dim
D^\perp = q^m-1 - m\lceil (\delta-1)(1-1/q)\rceil=n-k$, whence $k=\dim
D$. Let $T_D$ and $T_{D^\perp}$ respectively
denote the defining sets of $D$ and $D^\perp$.
It follows from the definitions that
$T_{D^\perp}= \bigcup_{a=1}^{\delta-1} C_a$ and that $T$ is a subset of
$$N\setminus (T_{D^\perp} \cup T_{D^\perp}^{-1}) = (N\setminus
T_{D^{\perp}}^{-1})\setminus T_{D^\perp}= T_{D} \setminus T_{D^\perp}.$$
If $T_C= T_D\setminus (T\cup T^{-1})$
denotes the defining set of a cyclic code $C$, then $\dim C = k+r$. By
Proposition~\ref{lem:cyclic-subsysI}, there exists an
$[[n,n-2k-r,r,\ge \delta]]_q$ subsystem code, which proves the claim.
\end{proof}

\begin{table}[t]
\caption{subsystem BCH codes that are derived using the Euclidean
construction}
\label{table:bchtable}
\begin{center}
\begin{tabular}{|l|l|c|}
\hline
\text{Subsystem Code} &  \text{Parent BCH} & \text {Designed}  \\
 &  \text{Code} $C$ & \text{distance }  \\
 \hline
 &&\\
 $[[15 ,4 ,3 ,3 ]]_2$   &$[15 ,7 ,5 ]_2$  & 4\\
 $ [[15 ,6 ,1 ,3 ]]_2  $ &$[15 ,5 ,7 ]_2 $ & 6\\
  $ [[31 ,10,1 ,5 ]]_2 $  &$[31 ,11,11]_2 $ & 8\\
  $  [[31 ,20,1 ,3 ]]_2  $ &$[31 ,6 ,15]_2 $ & 12\\
   $  [[63 ,6 ,21,7 ]]_2 $  &$[63 ,39,9 ]_2 $ & 8\\
$ [[63 ,6 ,15,7 ]]_2 $  &$[63 ,36,11]_2$  & 10\\
 $ [[63 ,6 ,3 ,7 ]]_2 $  &$[63 ,30,13]_2$  & 12\\
$ [[63 ,18,3 ,7 ]]_2$   &$[63 ,24,15]_2$  & 14\\
$  [[63 ,30,3 ,5 ]]_2$   &$[63 ,18,21]_2 $ & 16\\
  $ [[63 ,32,1 ,5 ]]_2 $  &$[63 ,16,23]_2 $ & 22\\
  $  [[63 ,44,1 ,3 ]]_2  $ &$[63 ,10,27]_2 $ & 24\\
  $  [[63 ,50,1 ,3 ]]_2  $ &$[63 ,7 ,31]_2  $& 28\\
  \hline
  &&\\
  $[[15 ,2 ,5 ,3 ]]_4$   &$[15 ,9 ,5 ]_4$  & 4\\
   $[[15 ,2 ,3 ,3 ]]_4  $ &$[15 ,8 ,6 ]_4$  & 6\\
    $[[15 ,4 ,1 ,3 ]]_4  $ &$[15 ,6 ,7 ]_4$  & 7\\
     $[[15 ,8 ,1 ,3 ]]_4  $ &$[15 ,4 ,10]_4$  & 8\\

$[[31 ,10,1 ,5 ]]_4  $ &$[31 ,11,11]_4$  & 8\\
 $[[31 ,20,1 ,3 ]]_4  $ &$[31 ,6 ,15]_4$  & 12\\
  $[[63 ,12,9 ,7 ]]_4  $ &$[63 ,30,15]_4$  & 15\\
   $[[63 ,18,9 ,7 ]]_4  $ &$[63 ,27,21]_4$  & 16\\
    $[[63 ,18,7 ,7 ]]_4  $ &$[63 ,26,22]_4$  & 22\\

 \hline
\end{tabular}
\\$*$ punctured code\\
$+$ Extended code
\end{center}
\end{table}

Similarly, we can obtain a hermitian variation of the preceding
corollary with the help of Proposition~\ref{lem:cyclic-subsysII}.
\begin{corollary}
Consider a primitive, narrow-sense BCH
code of length $n=q^{2m}-1$ with $m\neq 2$ over $\F_q$ with designed
distance $\delta$ in the range
\begin{equation}\label{eq:ddistrange2}
2\le \delta\le q^m-1
\end{equation}
If $T$ is a subset of the set $N \setminus \left(\bigcup_{a=1}^{\delta-1} (C_a\cup
C_{-qa})\right)$ that is a union
of cyclotomic cosets and $r=|T\cup T^{-q}|$ with $0\le r< n-2k$, where
$k=m\lceil (\delta-1)(1-1/{q^2})\rceil$,  then there exists a
$$ [[q^{2m}-1,q^{2m}-1-
2m\lceil (\delta-1)(1-1/{q^2})\rceil -r,r,\ge \delta]]_q$$
subsystem code.
\end{corollary}
\begin{proof}
The proof is similar to the proof of the previous corollary, and is a
consequence of~\cite[Theorems~4 and 7]{aly06a} and
Proposition~\ref{lem:cyclic-subsysII}.
\end{proof}

It is straightforward to generalize the previous two corollaries to
the case of non-primitive BCH codes using the results given
in~\cite{aly07a,aly08phd}.

One of the disadvantages of the cyclic constructions is that the
parameter $r$ is restricted to values dictated by the possible
cardinalities of the sets $T\cup T^{-1}$ (or $T\cup T^{-q}$), where
$T$ is confined to be a union of cyclotomic cosets. In the next
section, we will see how one can overcome this limitation.

We conclude this section by giving some examples of the parameters of
subsystem BCH codes in Tables~\ref{table:bchtable}
and~\ref{table:bchtableII}.

\begin{table}[t]
\caption{Subsystem BCH codes that are derived with the help of the hermitian construction}
\label{table:bchtableII}
\begin{center}
\begin{tabular}{|c|c|c|}
\hline
\text{Subsystem Code} &  \text{Parent BCH} & \text {Designed}  \\
 &  \text{Code $C$} & \text{ distance }  \\
 \hline
 $[[14 ,1 ,3 ,4 ]]_2$ &$[14 ,8 ,5 ]_{2^2}$&$6^*$ \\
$[[15,1,2,5]]_2$ & $[15,8,6]_{2^2}$ &6 \\
{} $[[15,5,2,3]]_2$&$[15,6,7]_{2^2}$&7\\
$[[16 ,5 ,2 ,3 ]]_2$ &$ [16 ,6 ,7 ]_{2^2}$&$7^+$ \\
 \hline
$[[17,8,1,4]]_2 $&$ [17,5,9]_{2^2}$ &4 \\
\hline
 $[[21,6,3,3]]_2$&$ [21,9,7]]_{2^2}$&6\\{}
 $[[21 ,7 ,2 ,3 ]]_2$& $ [21 ,8 ,9 ]_{2^2}$&8\\
\hline $[[31,10,1,5]]_2$&$[31,11,11]_{2^2} $&8\\{}
$[[31 ,20,1 ,3 ]]_2$&$ [31 ,6 ,15]_{2^2}$&12\\
$[[32 ,10,1 ,5 ]]_2$ &$ [32 ,11,11]_{2^2}$&$8^+$\\
$[[32 ,20,1 ,3 ]]_2$ &$[32 ,6 ,15]_{2^2}$&$12^+$\\
 \hline $[[25 ,12,3 ,3 ]]_3$ & $[25 ,8
,12]_{3^2}$&$9^*$
\\
$[[26 ,6 ,2 ,5 ]]_3$&$[26 ,11,8 ]_{3^2}$&8\\
$[[26 ,12,2 ,4 ]]_3 $&$[26 ,8 ,13]_{3^2} $&9\\
$[[26 ,13,1 ,4 ]]_3$&$[26 ,7 ,14]_{3^2}$&14\\
\hline
$[[80 ,1 ,17,20]]_3$&$[80 ,48,21]_{3^2} $&21\\
$[[80 ,5 ,17,17]]_3 $& $[80 ,46,22]_{3^2}$&22\\
 \hline
\end{tabular}
\\$*$ punctured code\\
$+$ Extended code
\end{center}
\end{table}

\section{Trading Dimensions of subsystem and co-subsystem codes}\label{sec:dimensions}
In this section we show how one can trade the dimensions of subsystem and
co-subsystem to obtain new codes from a given subsystem or stabilizer code.
The results are obtained by exploiting the symplectic geometry of the space. A
remarkable consequence is that nearly any stabilizer code yields a series of
subsystem codes.

Our first result shows that one can decrease the dimension of the
subsystem and increase at the same time the dimension of the
co-subsystem while keeping or increasing the minimum distance of the
subsystem code.

\begin{theorem}\label{th:shrinkK}
Let $q$ be a power of a prime~$p$. If there exists an $((n,K,R,d))_q$
subsystem code with $K>p$ that is pure to $d'$, then there exists an
$((n,K/p,pR,\geq d))_q$ subsystem code that is pure to $\min\{d,d'\}$.
If a pure $((n,p,R,d))_q$ subsystem code exists, then there exists a
$((n,1,pR,d))_q$ subsystem code.
\end{theorem}
\begin{proof}
By definition, an $((n,K,R,d))_q$ Clifford subsystem code is
associated with a classical additive code $C \subseteq \F_q^{2n}$ and
its subcode $D=C\cap C^\sdual$ such that $x=|C|$, $y=|D|$,
$K=q^n/(xy)^{1/2}$, $R=(x/y)^{1/2}$, and $d=\swt(D^\sdual - C)$ if
$C\neq D^\sdual$, otherwise $d=\swt(D^\sdual)$ if $D^\sdual=C$.

We have $q=p^m$ for some positive integer $m$. Since $K$ and $R$ are
positive integers, we have $x=p^{s+2r}$ and $y=p^s$ for some integers
$r\ge 1$, and $s\ge 0$. There exists an $\F_p$-basis of $C$
of the form
$$ C = \spann_{\F_p}\{z_1,\dots,z_s,x_{s+1},z_{s+1},\dots,
x_{s+r},z_{s+r}\}$$ that can be extended to a symplectic basis
$\{x_1,z_1,\dots,x_{nm},z_{nm}\}$ of $\F_q^{2n}$, that is,
$\scal{x_k}{x_\ell}=0$, $\scal{z_k}{z_\ell}=0$,
$\scal{x_k}{z_\ell}=\delta_{k,\ell}$ for all $1\le k,\ell \le nm$,
see~\cite[Theorem 8.10.1]{cohn05}.

Define an additive code $$C_m =
\spann_{\F_p}\{z_1,\dots,z_s,x_{s+1},z_{s+1},\dots,
x_{s+r+1},z_{s+r+1}\}.$$ It follows that
$$C^\sdual_m=\spann_{\F_p}\{z_1,\dots,z_s,x_{s+r+2},z_{s+r+2}, \dots,
x_{nm},z_{nm}\}$$ and
$$D=C_m\cap C_m^\sdual =
\spann_{\F_p}\{z_1,\dots,z_s\}.$$
By definition, the code $C$ is a
subset of $C_m$.

The subsystem code defined by $C_m$ has the parameters
$(n,K_m,R_m,d_m)$, where $K_m=q^n/(p^{s+2r+2}p^s)^{1/2}=K/p$ and
$R_m=(p^{s+2r+2}/p^s)^{1/2}=pR$. For the claims concerning
minimum distance and purity, we distinguish two cases:
\begin{compactenum}[(a)]
\item If $C_m\neq D^\sdual$, then $K>p$ and $d_m=\swt(D^\sdual -
C_m)\ge \swt(D^\sdual-C)=d$. Since by hypothesis $\swt(D^\sdual-C)=d$
and $\swt(C)\ge d'$, and $D\subseteq C\subset C_m\subseteq D^\sdual$
by construction, we have $\swt(C_m)\ge \min\{ d,d'\}$; thus, the
subsystem code is pure to $\min\{d,d'\}$.

\item If $C_m=D^\sdual$, then $K_m=1=K/p$, that is, $K=p$;  it follows from
the assumed purity that $d=\swt(D^\sdual-C)=\swt(D^\sdual)=d_m$.
\end{compactenum}
This proves the claim.
\end{proof}

For $\F_q$-linear subsystem codes there exists a variation of the
previous theorem which asserts that one can construct the resulting
subsystem code such that it is again $\F_q$-linear.

\begin{theorem}\label{th:FqshrinkK}
Let $q$ be a power of a prime~$p$. If there exists an $\F_q$-linear
$[[n,k,r,d]]_q$ subsystem code with $k>1$ that is pure to $d'$, then
there exists an $\F_q$-linear $[[n,k-1,r+1,\geq d]]_q$ subsystem code
that is pure to $\min\{d,d'\}$.  If a pure $\F_q$-linear
$[[n,1,r,d]]_q$ subsystem code exists, then there exists an
$\F_q$-linear $[[n,0,r+1,d]]_q$ subsystem code.
\end{theorem}
\begin{proof}
The proof is analogous to the proof of the previous theorem, except
that $\F_q$-bases are used instead of $\F_p$-bases.
\end{proof}

There exists a partial converse of Theorem~\ref{th:shrinkK}, namely if
the subsystem code is pure, then it is possible to increase the
dimension of the subsystem and decrease the dimension of the
co-subsystem while maintaining the same minimum distance.

\begin{theorem}\label{th:shrinkR}
Let $q$ be a power of a prime $p$. If there exists a pure
$((n,K,R,d))_q$ subsystem code with $R>1$, then there exists a pure
$((n,pK,R/p,d))_q$ subsystem code.
\end{theorem}
\begin{proof}
Suppose that the $((n,K,R,d))_q$ Clifford subsystem code is associated
with a classical additive code
$$ C_m = \spann_{\F_p}\{z_1,\dots,z_s,x_{s+1},z_{s+1},\dots,
x_{s+r+1},z_{s+r+1}\}.$$ Let $D=C_m\cap C_m^\sdual$. We have
$x=|C_m|=p^{s+2r+2}$, $y=|D|=p^s$, hence $K=q^n/p^{r+s}$ and
$R=p^{r+1}$. Furthermore, $d=\swt(D^\sdual)$.

The code $$C=\spann_{\F_p}\{z_1,\dots,z_s,x_{s+1},z_{s+1},\dots,
x_{s+r},z_{s+r}\}$$ has the subcode $D=C\cap C^\sdual$. Since
$|C|=|C_m|/p^2$, the parameters of the Clifford subsystem code
associated with $C$ are $((n,pK,R/p,d'))_q$. Since $C\subset C_m$, the
minimum distance $d'$ satisfies $$d'=\swt(D^\sdual-C)\le \swt(D^\sdual
- C_m)=\swt(D^\sdual)=d.$$  On the other hand, $d'=\swt(D^\sdual-C)\ge
\swt(D^\sdual)=d$, whence $d=d'$. Furthermore, the resulting code is pure
since $d=\swt(D^\sdual)=\swt(D^\sdual-C)$.
\end{proof}

Replacing $\F_p$-bases by $\F_q$-bases in the proof of the
previous theorem yields the following variation of the previous
theorem for $\F_q$-linear subsystem codes.
\begin{theorem}\label{th:FqshrinkR}
Let $q$ be a power of a prime $p$. If there exists a pure
$\F_q$-linear $[[n,k,r,d]]_q$ subsystem code with $r>0$, then there
exists a pure $\F_q$-linear $[[n,k+1,r-1,d]]_q$ subsystem code.
\end{theorem}

The purity hypothesis in Theorems~\ref{th:shrinkR}
and~\ref{th:FqshrinkR} is essential, as the next remark shows.

\begin{remark}
The Bacon-Shor code is an impure $[[9,1,4,3]]_2$ subsystem
code. However, there does not exist any $[[9,5,3]]_2$ stabilizer code.
Thus, in general one cannot omit the purity assumption from
Theorems~\ref{th:shrinkR} and~\ref{th:FqshrinkR}.
\end{remark}

An $[[n,k,d]]_q$ stabilizer code can also be regarded as an
$[[n,k,0,d]]_q$ subsystem code. We record this important special case
of the previous theorems in the next corollary.

\goodbreak
\begin{corollary}\label{cor:generic}
If there exists an ($\F_q$-linear) $[[n,k,d]]_q$ stabilizer code that
is pure to $d'$, then there exists for all $r$ in the range $0\le r<k$
an ($\F_q$-linear) $[[n,k-r,r,\ge d]]_q$ subsystem code that is pure
to $\min\{d,d'\}$ .  If a pure ($\F_q$-linear) $[[n,k,r,d]]_q$
subsystem code exists, then a pure ($\F_q$-linear) $[[n,k+r,d]]_q$
stabilizer code exists.
\end{corollary}

We have shown in~\cite{aly07a,aly06a} that a (primitive or non-primitive)
narrow sense BCH code of length $n$ over $\F_q$
contains its dual code if
the designed distance $\delta$ is in the range
$$2\le \delta\le \delta_{\max}=\frac{n}{q^{m}-1} (q^{\lceil
m/2\rceil}-1-(q-2)[m \textup{ odd}]).$$  For simplicity, we will
proceed our work for primitive narrow sense BCH codes, however, the
generalization for non-primitive BCH codes is a straightforward.

\begin{corollary}\label{lem:BCHExistFq}
If $q$ is  power of a prime, $m$ is a positive integer, and $2\leq
\delta\leq q^{\lceil m/2\rceil}-1 -(q-2)[m \text{ odd }]$. Then
there exists a subsystem BCH code with parameters
$[[q^m-1,n-2m\lceil(\delta-1)(1-1/q) \rceil -r, r,\geq \delta ]]_q$
where $0 \leq r< n-2m\lceil(\delta-1)(1-1/q) \rceil$.
\end{corollary}
\begin{proof}
We know that if $2\leq \delta\leq q^{\lceil m/2\rceil}-1 -(q-2)[m
\text{ odd }]$, then there exists a stabilizer code with parameters
$[[q^m-1,n-2m\lceil(\delta-1)(1-1/q) \rceil, \geq \delta ]]_q$. Let
r be an integer in the range $0 \leq r< n-2m\lceil(\delta-1)(1-1/q)
\rceil$. From~\cite[Theorem 2]{aly08a}, then there must exist a
subsystem BCH code with parameters $
[[q^m-1,n-2m\lceil(\delta-1)(1-1/q) \rceil -r, r,\geq \delta ]]_q$.
\end{proof}

\medskip

We can also construct subsystem BCH codes from stabilizer codes
using the Hermitian constructions.

\begin{corollary}\label{lem:BCHExistFq2}
If $q$ is a power of a prime, $m$ is a positive integer, and
$\delta$ is an integer in the range $2\le \delta \le
\delta_{\max}=q^{m+[m \textup{ even}]}-1 -(q^2-2)[m \textup{
even}]$, then there exists a subsystem code $Q$ with parameters
$$ [[q^{2m}-1, q^{2m}-1-2m\lceil(\delta-1)(1-1/q^2)\rceil -r,r, d_Q\ge
\delta]]_q$$ that is pure up to $\delta$, where $0 \leq r
<q^{2m}-1-2m\lceil(\delta-1)(1-1/q^2)\rceil$.
\end{corollary}
\begin{proof}
If $2\le \delta \le \delta_{\max}=q^{m+[m \textup{ even}]}-1
-(q^2-2)[m \textup{ even}]$, then exists a classical BCH code with
parameters $[q^m-1,q^m-1-m\lceil(\delta-1)(1-1/q)\rceil,\ge
\delta]_q$ which contains its dual code. From~\cite[Theorem
2]{aly08a}, then there must exist a subsystem code with the given
parameters.
\end{proof}

\section{MDS Subsystem Codes}\label{sec:MDSsubsys}
Recall that an $[[n,k,r,d]]_q$ subsystem code derived from an
$\F_q$-linear classical code $C\le \F_q^{2n}$ satisfies the Singleton
bound $k+r\le n-2d+2$, see~\cite[Theorem~3.6]{pre0703}. A subsystem
code attaining the Singleton bound with equality is called an MDS
subsystem code.

An important consequence of the previous theorems is the following
simple observation which yields an easy construction of subsystem codes
that are optimal among the $\F_q$-linear Clifford subsystem codes.

\begin{theorem}\label{th:pureMDS}
If there exists an $\F_q$-linear $[[n,k,d]]_q$ MDS stabilizer code,
then there exists a pure $\F_q$-linear $[[n,k-r,r,d]]_q$ MDS subsystem
code for all $r$ in the range $0\le r\le k$.
\end{theorem}
\begin{proof}
An MDS stabilizer code must be pure, see~\cite[Theorem~2]{rains99}
or \cite[Corollary 60]{ketkar06}. By Corollary~\ref{cor:generic}, a
pure $\F_q$-linear $[[n,k,d]]_q$ stabilizer code implies the
existence of an $\F_q$-linear $[[n,k-r,r, d_r\ge d]]_q$ subsystem
code that is pure to~$d$ for any $r$ in the range $0\le r\le k$.
Since the stabilizer code is MDS, we have $k=n-2d+2$. By the
Singleton bound, the parameters of the resulting $\F_q$-linear
$[[n,n-2d+2-r,r,d_r]]_q$ subsystem codes must satisfy
$(n-2d+2-r)+r\le n-2d_r+2$, which shows that the minimum distance
$d_r=d$, as claimed.
\end{proof}

\begin{remark}
We conjecture that $\F_q$-linear MDS subsystem codes are actually
optimal among all subsystem codes, but a proof that the Singleton
bound holds for general subsystem codes remains elusive.
\end{remark}
\smallskip

In the next corollary, we give a few examples of MDS subsystem codes
that can be obtained from Theorem~\ref{th:pureMDS}. These are the
first families of MDS subsystem codes (though sporadic examples of MDS
subsystem codes have been established before, see e.g.~\cite{aly06c,bacon06}).
\begin{corollary}
\begin{enumerate}[i)]
\item An $\F_q$-linear pure $[[n,n-2d+2-r,r,d]]_q$ MDS subsystem code exists
for all $n$, $d$, and $r$ such that $3\le n\le q$, $1\le d\le n/2+1$,
and\/ $0\le r\le n-2d+1$.
\item An $\F_q$-linear pure $[[(\nu+1)q,(\nu+1)q-2\nu-2-r,r,\nu+2]]_q$ MDS subsystem code exists for all $\nu$ and $r$ such that $0\le \nu\le q-2$ and
$0\le r\le (\nu+1)q-2\nu-3$.
\item An $\F_q$-linear pure $[[q - 1, q-1 -2\delta -r,
r,\delta + 1]]_q$ MDS subsystem code exists for all
$\delta$ and $r$ such that $0 \leq \delta < (q -1)/2$ and $0\leq r \le q - 2\delta - 1$.
\item An $\F_q$-linear pure $[[q, q -
2\delta - 2-r',r', \delta + 2]]_q$ MDS subsystem code exists for all
$0 \leq \delta < (q -1)/2$ and $0\leq r' <q - 2\delta - 2$.
\item An $\F_q$-linear pure $[[q^2 - 1, q^2 - 2\delta - 1-r,r, \delta +
1]]_q$ MDS subsystem code exists for all $\delta$ and $r$ in the range
$0 \leq \delta < q-1$ and $0\leq r< q^2 - 2\delta - 1$.
\item An $\F_q$-linear pure $[[q^2, q^2 - 2\delta -
2-r',r', \delta + 2]]_q$ MDS subsystem code exists for all $\delta$ and $r'$ in the range
$0 \leq \delta < q-1$ and $0\leq r' <q^2 - 2\delta - 2$.
\end{enumerate}
\end{corollary}
\begin{proof}
\begin{inparaenum}[i)]
\item By \cite[Theorem~14]{grassl04}, there exist $\F_q$-linear
$[[n,n-2d+2,d]]_q$ stabilizer codes for all $n$ and $d$ such that
$3\le n\le q$ and $1\le d\le n/2+1$. The claim follows from Theorem~\ref{th:pureMDS}. \\
\item By \cite[Theorem~5]{klappenecker050}, there exist a
$[[(\nu+1)q,(\nu+1)q-2\nu-2,\nu+2]]_q$ stabilizer code. In this case,
the code is derived from an $\F_{q^2}$-linear code $X$ of length $n$
over $\F_{q^2}$ such that $X\subseteq X^\hdual$. The claim follows
from Lemma~\ref{l:hermitian-linear} and Theorem~\ref{th:pureMDS}.\\
\item$\!\!$,\item There exist $\F_q$-linear stabilizer codes with
parameters $[[q - 1, q - 2\delta - 1,\delta + 1]]_q$ and $[[q, q -
2\delta - 2, \delta + 2]]_q$ for $0 \leq \delta < (q -1)/2$, \
see~\cite[Theorem~9]{grassl04}. Theorem~\ref{th:pureMDS} yields the claim. \\
\item$\!\!$,\item There exist $\F_q$-linear stabilizer codes with
parameters $[[q^2 - 1, q^2 - 2\delta - 1, \delta + 1]]_q$ and $[[q^2,
q^2 - 2\delta - 2, \delta + 2]]_q$.  for $0 \leq \delta < q-1$
by~\cite[Theorem~10]{grassl04}. The claim follows from
Theorem~\ref{th:pureMDS}.
\end{inparaenum}
\end{proof}
The existence of the codes in i) are merely established by a
non-constructive Gilbert-Varshamov type counting argument.  However,
the result is interesting, as it asserts that there exist for example
$[[6,1,1,3]]_q$ subsystem codes for all prime powers $q\ge 7$,
$[[7,1,2,3]]_q$ subsystem codes for all prime powers $q\ge 7$, and
other short subsystem codes that one should compare with a
$[[5,1,3]]_q$ stabilizer code. If the syndrome calculation is simpler,
then such subsystem codes could be of practical value.

The subsystem codes given in ii)-vi) of the previous corollary are
constructively established. The subsystem codes in ii) are derived
from Reed-Muller codes, and in iii)-vi) from Reed-Solomon codes.
There exists an overlap between the parameters given in ii) and in
iv), but we list here both, since each code construction has its own
merits.
\begin{remark}
By Theorem~\ref{th:FqshrinkR}, pure MDS subsystem codes can always
be derived from MDS stabilizer codes, see Table~\ref{table:optimalMDS}. Therefore,
one can derive in fact all possible parameter sets of pure MDS
subsystem codes with the help of Theorem~\ref{th:pureMDS}.
\end{remark}
\begin{remark}
In the case of stabilizer codes, all MDS codes must be pure. For
subsystem codes this is not true, as the $[[9,1,4,3]]_2$ subsystem
code shows. Finding such impure $\F_q$-linear $[[n,k,r,d]]_q$ MDS
subsystem codes with $k+r=n-2d+2$ is a particularly interesting
challenge.
\end{remark}

\begin{table}[t]
\caption{Optimal pure subsystem codes} \label{table:optimalMDS}
\begin{center}
\begin{tabular}{|c|c|c|}
\hline
\text{Subsystem Codes} &  \text{Parent}  \\
 &  \text{Code (RS Code)}  \\
\hline $[[8  ,1  ,5  ,2  ]]_3$ & $[8  ,6  ,3  ]_{3^2}$ \\{}

$[[8  ,4  ,2  ,2  ]]_3$&$[8  ,3  ,6  ]_{3^2}$\\{} $[[8  ,5  ,1  ,2
]]_3$&$[8 ,2 ,7 ]_{3^2}$\\{}

$[[9  ,1 ,4 ,3 ]]_3$&$[9  ,6 ,4 ]_{3^2}^{\dag}, \delta=3$\\
$[[9  ,4 ,1 ,3 ]]_3$&$[9  ,3 ,7 ]_{3^2}^{\dag}, \delta=6$ \\
 \hline $[[15,1,10,3]]_4$ &
$[15 ,12 ,4 ]_{4^2}$  \\{}
$[[15,9,2,3]]_4$&$[15,4,12]_{4^2}$\\{} $[[15,10,1,3]]_4$&$[15,3,13]_{4^2}$\\
$[[16 ,1 ,9 ,4 ]]_4$& $[16 ,12,5 ]_{4^2}^{\dag},\delta= 4 $\\

  \hline $[[24,1,17,4]]_5$
&$[24,20,5]_{5^2}$
\\{}
$ [[24,16,2,4]]_5$ &$[24,5,20]_{5^2}$\\{} $[[24,17 ,1,4 ]]_5
$&$[24,4,21]_{5^2}$\\{}
 $[[24,19,1,3]]_5$ &$[24,3,22]_{5^2}$\\{}
$[[24 ,21 ,1  ,2  ]]_5$ & $ [24 ,2  ,23 ]_{5^2}$ \\{} $[[23 ,1 ,18,3
]]_5$&$[23 ,20,4 ]_{5^2}^{*}, \delta=5$\\{}
$[[23 ,16,3 ,3 ]]_5$&$[23 ,5 ,19]_{5^2}^{*}, \delta=20$\\
 \hline
 $[[48 ,1  ,37 ,6  ]]_7$  &$[48 ,42 ,7 ]_{7^2}$\\
 \hline
\end{tabular}
\\
* Punctured code\\
$\dag$ Extended code
\end{center}
 \end{table}

Recall that a pure subsystem code is called perfect if and only if it
attains the Hamming bound with equality. We conclude this section with
the following consequence of Theorem~\ref{th:pureMDS}:
\begin{corollary}
If there exists an $\F_q$-linear pure $[[n,k,d]]_q$ stabilizer code
that is perfect, then there exists a pure $\F_q$-linear
$[[n,k-r,r,d]]_q$ perfect subsystem code for all $r$ in the range
$0\leq r \leq k$.
\end{corollary}

\section{Extending and Shortening Subsystem Codes}\label{sec:extendshortensubsys}
In Section~\ref{sec:dimensions}, we showed how one can derive new
subsystem codes from known ones by modifying the dimension of the
subsystem and co-subsystem. In this section, we derive new subsystem
codes from known ones by extending and shortening the length of the
code.

\begin{theorem}\label{lemma_n+1k}
If there exists an $((n,K,R,d))_q$ Clifford subsystem  code
with $K>1$, then there exists an $((n+1, K, R, \ge d))_q$
subsystem code that is pure to~1.
\end{theorem}
\begin{proof}
We first note that for any additive subcode $X\le \F_q^{2n}$, we can
define an additive code
$X'\le \F_q^{2n+2}$ by
$$X'=\{ (a\alpha|b0)\,|\, (a|b)\in X, \alpha\in
\F_q\}.$$ We have $|X'|=q|X|$. Furthermore, if $(c|d)\in X^\sdual$,
then $(c\alpha|d0)$ is contained in $(X')^\sdual$ for all $\alpha$ in
$\F_q$, whence $(X^\sdual)'\subseteq (X')^\sdual$. By comparing
cardinalities we find that equality must hold; in other words, we have
$$(X^\sdual)'= (X')^\sdual.$$

By Theorem~\ref{th:oqecfq}, there are two additive codes $C$ and $D$
associated with an $((n,K,R,d))_q$ Clifford subsystem code such that
$$|C|=q^nR/K$$ and $$|D|=|C\cap C^\sdual| = q^n/(KR).$$ We can derive from the
code $C$ two new additive codes of length $2n+2$ over $\F_q$, namely $C'$ and
$D'=C'\cap (C')^\sdual$. The codes $C'$ and $D'$ determine a
$((n+1,K',R',d'))_q$ Clifford subsystem code. Since
\begin{eqnarray*}
D'&=&C'\cap (C')^\sdual = C'\cap (C^\sdual)' \\&=&(C\cap C^\sdual)',
\end{eqnarray*}
 we have
$|D'|=q|D|$. Furthermore, we have $|C'|=q|C|$. It follows from
Theorem~\ref{th:oqecfq} that
\begin{compactenum}[(i)]
\item $K'= q^{n+1}/\sqrt{|C'||D'|}=q^n/\sqrt{|C||D|}=K$,
\item $R'=(|C'|/|D'|)^{1/2} = (|C|/|D|)^{1/2} = R$,
\item $d'= \swt( (D')^\sdual \setminus C')\ge \swt( (D^\sdual\setminus C)')=d$.
\end{compactenum}
Since $C'$ contains a vector $(\mathbf{0}\alpha|\mathbf{0}0)$ of
weight $1$, the resulting subsystem code is pure to~1.
\end{proof}


\begin{corollary}
If there exists an $[[n,k,r,d]]_q$ subsystem  code with $k>0$ and $0\leq r <k$,
then there exists an $[[n+1, k, r, \ge d]]_q$ subsystem code that is pure to~1.
\end{corollary}
\medskip

We can also shorten the length of a subsystem code in a simple way
as shown in the following Theorem.

\begin{theorem}\label{lem:n-1k+1rule}
If a pure $((n,K,R,d))_q$ subsystem code exists, then there exists
a pure $((n-1,qK,R,d-1))_q$ subsystem code.
\end{theorem}
\begin{proof}
By \cite[Lemma~10]{aly06c}, the existence of a pure Clifford subsystem
code with parameters $((n,K,R,d))_q$ implies the existence of a pure
$((n,KR,d))_q$ stabilizer code. It follows from~\cite[Lemma
70]{ketkar06} that there exist a pure $((n-1,qKR,d-1))_q$ stabilizer
code, which can be regarded as a pure $((n-1,qKR,1,d-1))_q$ subsystem
code. Thus, there exists a pure $((n-1,qK,R,d-1))_q$ subsystem code by
Theorem~\ref{th:shrinkR}, which proves the claim.
\end{proof}

In bracket notation, the previous theorem states that the existence of
a pure $[[n,k,r,d]]_q$ subsystem code implies the existence of a pure
$[[n-1,k+1,r,d-1]]_q$ subsystem code.

\section{Combining Subsystem Codes}\label{sec:combinesubsys}
In this section, we show how one can obtain a new subsystem code by
combining two given subsystem codes in various ways.

\begin{theorem}\label{thm:twocodes_n1k1r1d1n2k2r2d2}
If there exists a pure $[[n_1,k_1,r_1,d_1]]_2$ subsystem code and a
pure $[[n_2,k_2,r_2,d_2]]_2$ subsystem code such that $k_2+r_2\leq n_1$,
then there exist subsystem codes with parameters
$$[[n_1+n_2-k_2-r_2,k_1+r_1-r,r,d]]_2$$
for all $r$ in the range $0\le r< k_1+r_1$, where
the minimum distance $d \geq \min\{d_1,d_1+d_2-k_2-r_2\}$.
\end{theorem}
\begin{proof}
Since there exist pure $[[n_1,k_1,r_1,d_1]]_2$ and
$[[n_2,k_2,r_2,d_2]]_2$ subsystem codes with $k_2+r_2\leq n_1$, it follows
from Theorem~\ref{th:shrinkR} that there exist stabilizer codes with
the parameters $[[n_1,k_1+r_1,d_1]]_2$ and $[[n_2,k_2+r_2,d_2]]_2$
such that $k_2+r_2\leq n_1$. Therefore, there exists an
$[[n_1+n_2-k_2-r_2,k_1+r_1,d]]_2$ stabilizer code with minimum distance
$$d \geq \min
\{d_1,d_1+d_2-k_2-r_2\}$$ by \cite[Theorem 8]{calderbank98}.
It follows from Theorem~\ref{th:shrinkK} that there exists
$[[n_1+n_2-k_2-r_2,k_1+r_1-r,r,\geq d]]_2$ subsystem codes for all
$r$ in the range $0\le r< k_1+r_1$.
\end{proof}

\begin{theorem}\label{lem:twocodes_nk1r1s1k2r2d2} Let $Q_1$ and $Q_2$ be two
pure subsystem codes with parameters $[[n,k_1,r_1,d_1]]_q$ and
$[[n,k_2,r_2,d_2]]_q$, respectively. If $Q_2\subseteq Q_1$, then
there exists pure subsystem codes with parameters
$$[[2n,k_1+k_2+r_1+r_2-r,r,d]]_q$$
for all $r$ in the range $0\le r\le k_1+k_2+r_1+r_2$, where
the minimum distance $d \geq \min \{d_1,2d_2\}$.
\end{theorem}
\begin{proof}
By assumption, there exists a pure $[[n,k_i,r_i,d_i]]_q$ subsystem
code, which implies the existence of a pure $[[n,k_i+r_i,d_i]]_q$
stabilizer code by Theorem~\ref{th:shrinkR}, where $i\in\{1,2\}$.
By~\cite[Lemma 74]{ketkar06}, there exists a pure stabilizer code with
parameters $[[2n,k_1+k_2+r_1+r_2,d]]_q$ such that $d \geq \min
\{2d_2,d_1\}$. By Theorem~\ref{th:shrinkK}, there exist a pure
subsystem code with parameters $[[2n,k_1+k_2+r_1+r_2-r,r,d]]_q$ for
all $r$ in the range $0\le r\le k_1+k_2+r_1+r_2$, which proves the claim.
\end{proof}

Further analysis of propagation rules of subsystem code constructions, tables of upper and lower bounds, and short subsystem codes are presented in~\cite{aly08phd}.

\bigskip

\section{Conclusion and Discussion}\label{sec:conclusion}
 Subsystem codes are among
the most versatile tools in quantum error-correction, since they allow one
to combine the passive error-correction found in decoherence free
subspaces and noiseless subsystems with the active error-control methods
of quantum error-correcting codes.  In this paper we demonstrate several
methods of subsystem code constructions over binary and nonbinary fields.
The subclass of Clifford subsystem codes that was studied in this paper is
of particular interest because of the close connection to classical
error-correcting codes. As Theorem~\ref{th:oqecfq} shows, one can derive
from each additive code over $\F_q$ an Clifford subsystem code. This
offers more flexibility than the slightly rigid framework of stabilizer
codes.

We showed that any $\F_q$-linear MDS stabilizer code yields a series of
pure $\F_q$-linear MDS subsystem codes. These codes are known to be
optimal among the $\F_q$-linear Clifford subsystem codes. We conjecture
that the Singleton bound holds in general for subsystem codes. There is
quite some evidence for this fact, as pure Clifford subsystem codes and
$\F_q$-linear Clifford subsystem codes are known to obey this bound. 
We
have established a number of subsystem code constructions. In particular,
we have shown how one can derive subsystem codes from stabilizer codes. In
combination with the propagation rules that we have derived, one can
easily create tables with the best known subsystem codes. Further
propagation rules and examples of such tables are given in~\cite{aly08phd}, and will appear in an expanded
version of this paper.

\bigskip

\section{ACKNOWLEDGMENTS} This research was supported by NSF grant
CCF-0622201 and NSF CAREER award CCF-0347310. Part of this paper is appeared in Proceedings of 2008 IEEE International Symposium on Information Theory, ISIT'08, Toronto, CA, July 2008.

\bigskip

\scriptsize
\newcommand{\XXstud}{{}}
\newcommand{\XXar}[1]{}
\bibliographystyle{plain}
%
%

%
%

\bigskip

\appendix We recall that the Hermitian construction of stabilizer
codes yields $\F_q$-linear stabilizer codes, as can be seen from the
following reformulation of~\cite[Corollary~2]{grassl04}.
\begin{lemma}[\cite{grassl04}]\label{l:hermitian-linear}
If there exists an $\F_{q^2}$-linear code $X\subseteq \F_{q^2}^n$ such
that $X\subseteq X^\hdual$, then there exists an $\F_q$-linear code
$C\subseteq \F_q^{2n}$ such that $C\subseteq C^\sdual$, $|C|=|X|$,
$\swt(C^\sdual - C)=\wt(X^\hdual - X)$ and $\swt(C)=\wt(X)$.
\end{lemma}
\begin{proof}
Let $\{1,\beta\}$ be a basis of $\F_{q^2}/\F_q$. Then
$\tr_{q^2/q}(\beta)=\beta+\beta^q$ is an element $\beta_0$ of $\F_q$; hence,
$\beta^q=-\beta+\beta_0$. Let $$C=\{ (u|v)\,|\, u,v\in \F_q^n, u+\beta v\in
X\}.$$ It follows from this definition that $|X|=|C|$ and that
$\wt(X)=\swt(C)$. Furthermore, if $u+\beta v$ and $u'+\beta v'$ are elements of
$X$ with $u,v,u',v'$ in $\F_q^n$, then
$$
\begin{array}{lcl}
0&=&(u+\beta v)^q\cdot (u'+\beta v') \\
&=& u\cdot u' + \beta^{q+1} v\cdot
v' + \beta_0 v \cdot u' + \beta (u\cdot v' -v \cdot u').
\end{array}
$$ On the right hand side, all terms but the last are in $\F_q$; hence
we must have $(u\cdot v' -v \cdot u')=0$, which shows that $(u|v)
\,\sdual\, (u'|v')$, whence $C\subseteq C^\sdual$. Expanding
$X^\hdual$ in the basis $\{1\,\beta\}$ yields a code $C'\subseteq
C^\sdual$, and we must have equality by a dimension argument. Since
the basis expansion is isometric, it follows that $$\swt(C^\sdual -
C)=\wt(X^\hdual - X).$$ The $\F_q$-linearity of $C$ is a direct
consequence of the definition of $C$.
\end{proof}

\end{document}